\documentclass[aps,showpacs,english]{revtex4}
\usepackage{mathptmx}
\usepackage[T1]{fontenc}
\usepackage[latin9]{inputenc}
\usepackage{amsthm}
\usepackage{amsmath}
\usepackage{amssymb}

\makeatletter
\@ifundefined{textcolor}{}
{%
 \definecolor{BLACK}{gray}{0}
 \definecolor{WHITE}{gray}{1}
 \definecolor{RED}{rgb}{1,0,0}
 \definecolor{GREEN}{rgb}{0,1,0}
 \definecolor{BLUE}{rgb}{0,0,1}
 \definecolor{CYAN}{cmyk}{1,0,0,0}
 \definecolor{MAGENTA}{cmyk}{0,1,0,0}
 \definecolor{YELLOW}{cmyk}{0,0,1,0}
 }
\theoremstyle{plain}

\@ifundefined{definecolor}
 {\@ifundefined{definecolor}
 {\usepackage{color}}{}
}{}

\makeatother

\newtheorem{thm}{Theorem}  
 \newtheorem{lem}{Lemma}

\makeatother

\makeatother

\makeatother

\makeatother

\makeatother

\makeatother

\usepackage{babel}

\makeatother

\usepackage{babel}

\makeatother

\usepackage{babel}

\makeatother

\usepackage{babel}

\makeatother

\usepackage{babel}

\begin{document}

\title{Entanglement sharing through noisy qubit channels: One-shot optimal
singlet fraction }

\author{Rajarshi Pal }

\email{rajarshi@imsc.res.in}

\affiliation{Optics \& Quantum Information Group, The Institute of Mathematical
Sciences, C. I. T. Campus, Taramani, Chennai 600113, India}

\author{Somshubhro Bandyopadhyay}

\email{som@jcbose.ac.in}

\affiliation{Department of Physics and Center for Astroparticle Physics and Space
Science, Bose Institute, Block EN, Sector V, Bidhan Nagar, Kolkata
700091, India}

\author{Sibasish Ghosh}

\email{sibasish@imsc.res.in}

\affiliation{Optics \& Quantum Information Group, The Institute of Mathematical
Sciences, C. I. T. Campus, Taramani, Chennai 600113, India }
\begin{abstract}
Maximally entangled states--a resource for quantum information processing--can
only be shared through noiseless quantum channels, whereas in practice
channels are noisy. Here we ask: Given a noisy quantum channel, what
is the maximum attainable purity (measured by singlet fraction) of
shared entanglement for single channel use and local trace preserving
operations? We find an exact formula of the maximum singlet fraction
attainable for a qubit channel and give an explicit protocol to achieve
the optimal value. The protocol distinguishes between unital and nonunital
channels and requires no local post-processing. In particular, the
optimal singlet fraction is achieved by transmitting part of an appropriate
pure entangled state, which is maximally entangled if and only if
the channel is unital. A linear function of the optimal singlet fraction
is also shown to be an upper bound on the distillable entanglement
of the mixed state dual to the channel. 
\end{abstract}

\maketitle
\section{ Introduction:} \emph{Shared} entanglement between two separated
observers (Alice and Bob) is a critical resource for quantum information
processing (QIP) tasks such as dense coding \cite{dense-coding},
cryptography \cite{Ekert91}, distributed quantum computation \cite{distributed-comp},
and quantum teleportation \cite{teleportation}. Faithful implementation
of QIP tasks require maximally entangled states, which can only be
shared through noiseless quantum channels, where Alice prepares a
maximally entangled state of two particles (say, qubits) and sends
one of them to Bob through the channel. In practice, available channels
are noisy resulting in mixed states. Entanglement distillation \cite{distillation-1,distillation-2,distillation-3,Distillation04,fidelity}
provides a solution by converting these mixed states to fewer almost-perfect
entangled states of purity close to unity while requiring many uses
of the channel and joint measurements on many copies of the output.
Clearly, the yield in an entanglement distillation protocol depends
on the purity of the mixed states, which in turn is a function of
the amount of noise present in the quantum channel. Thus, in the simplest
case of entanglement sharing, a basic question is: Given a noisy quantum
channel what is the maximum achievable purity for single use of the
channel? 

In this work, we answer the above question for qubit channels within
the paradigm of trace-preserving local operations (TP-LOCC). By restricting
to this class of operations, where {no subsystem is thrown away,
our results provide the conditions and an explicit protocol when every
single use of the channel is maximally efficient. Our result also
characterizes qubit channels by quantifying reliable transmission
of quantum information via teleportation for single channel use and
TP-LOCC. 

In the simplest scenario, the general protocol of sharing entanglement
works as follows: Alice prepares a bipartite pure entangled state
$\vert\psi\rangle$ and sends one half of it to Bob through a quantum
channel, say $\Lambda$ (which, throughout the present paper,
is assumed to be nonentanglement breaking). This results, in general, in a mixed entangled
state $\rho_{\psi,\Lambda}=(I\otimes\Lambda)\rho_{\psi}$, where $\rho_{\psi}=\vert\psi\rangle\langle\psi\vert$.
The \emph{purity} of this state is characterized by its singlet fraction
\cite{fidelity,distillation-1,distillation-3,F<=1/2} defined as:
\begin{eqnarray}
F\left(\rho_{\psi,\Lambda}\right) & = & \max_{\vert\Phi\rangle}\langle\Phi\vert\rho_{\psi,\Lambda}\vert\Phi\rangle,\label{fidelity}\end{eqnarray}
 where $|\Phi\rangle$ is a maximally entangled state. The singlet
fraction quantifies how close the state $\rho_{\psi,\Lambda}$ is
to a maximally entangled state, and therefore how useful the state
is for QIP tasks. For example, it is related to the teleportation
fidelity $f$ for teleportation of a qudit via the following relation:
\begin{eqnarray}
f\left(\rho_{\psi,\Lambda}\right) & = & \frac{dF\left(\rho_{\psi,\Lambda}\right)+1}{d+1}\label{tel-fidelity}\end{eqnarray}

In this work we are interested in the \emph{optimal singlet fraction}
for the channel $\Lambda$ defined as : \begin{eqnarray}
F\left(\Lambda\right) & = & \max_{\vert\psi\rangle}\max_{\mbox{L}}F\left(L(\rho_{\psi,\Lambda})\right),\label{optimal-fidelity-0}\end{eqnarray}
 where the maximum is taken over all pure state transmissions and
trace preserving LOCCs $L$. Note that, by virtue of Eq.$\,$(\ref{tel-fidelity})
$F\left(\Lambda\right)$ also quantifies reliable transmission of
quantum states via teleportation, albeit for single channel use, where
the optimal teleportation fidelity for the channel is expressed as
$f\left(\Lambda\right)=\frac{dF\left(\Lambda\right)+1}{d+1}$. This
is in contrast with the known measures such as, channel fidelity \cite{fidelity},
which quantifies, on an average, how close the output state is to
the input state, and entanglement fidelity \cite{Schumacher-1996,BNS-98},
which captures how well the channel preserves entanglement \cite{entanglement}
of the transmitted system with other systems.

For qubit channels such as depolarizing \cite{fidelity} and amplitude
damping \cite{BG-2012} the value of $F\left(\Lambda\right)$ is known,
but no general expression has been found yet for a generic qubit channel.
In this work, we obtain an exact formula of $F\left(\Lambda\right)$
for a qubit channel and give an explicit protocol to achieve this
value. Surprisingly, we also find that to attain the optimal value
no local post processing is required, even though it is known that
local post-processing can increase the singlet fraction of a state.
In particular, we show that the optimal value is attained by sending
part of a maximally entangled state down the channel if and only if
the channel is unital. This means that for nonunital channels one
must necessarily transmit part of an appropriate nonmaximally entangled
state. We also prove that the optimal singlet fraction is equal to
a linear function of the negativity \cite{F<=1/2} of the mixed state $\rho_{\Phi^{+},\Lambda}$,
where $|\Phi^{+}\rangle=\frac{1}{\sqrt{2}}\left(|00\rangle+|11\rangle\right)$.
Thus a linear function of $F(\Lambda)$ is an upper bound on the distillable
entanglement of the mixed state $\rho_{\Phi^{+},\Lambda}$. 

Let us note a couple of implications of our results. As noted earlier,
an entanglement distillation \cite{distillation-1,distillation-2,distillation-3,Distillation04,fidelity}
protocol uses many copies of the mixed state $\rho_{\psi,\Lambda}$(for
some transmitted pure state $\vert\psi\rangle$) of purity $F\left(\rho_{\psi,\Lambda}\right)$
and converts them to a fewer number of near-perfect entangled states
of purity close to unity. Following the prescription in this paper,
for a given noisy qubit channel Alice and Bob can now prepare states
with maximum achievable purity for each channel use so as to maximize
the yield in their distillation protocol. Second, by virtue of Eq.$\,$(\ref{tel-fidelity})
we are able to provide the optimal teleportation fidelity for any
qubit channel, albeit for single channel use.

The paper is organized as follows: in section II we provide an analytical expression  for the optimal singlet fraction of any qubit channel and  a recipe for  obtaining the optimal value by sharing  a pure entangled state across
the channel. We also prove that this pure entangled state is maximally entangled if and only if the channel is unital. In section III we relate the optimal singlet fraction with the maximum output negativity of a state that can be shared across the channel.
In section IV we show that  for a non-unital qubit channel the singlet fraction obtained by post-processing the output of a maximally entangled state is strictly less than the optimal value. We conclude in section V.

\section{Optimal singlet fraction for qubit channels.}
\subsection{Preliminaries}
A quantum channel $\Lambda$ is a trace preserving
completely positive map characterized by a set of Kraus operators
$\left\{ A_{i}\right\} $ satisfying $\sum A_{i}^{\dagger}A_{i}=I$.
Its dual $\hat{\Lambda}$ is described in terms of the Kraus operators
$\left\{ A_{i}^{\dagger}\right\} $(the dual is the adjoint map with respect to the Hilbert-Schmidt inner product).  A channel $\Lambda$ is said
to be \emph{unital} if its action preserves Identity: $\Lambda\left(I\right)=I$,
and \emph{nonunital} if it does not, i.e., $\Lambda\left(I\right)\neq I$. A dual channel $\hat{\Lambda}$ is trace-preserving iff $\Lambda$ is unital. 
Sending half of pure state $\vert\phi\rangle$ down the channel $\$\in\left\{ \Lambda,\hat{\Lambda}\right\} $
gives rise to a mixed state \begin{eqnarray}
\rho_{\phi,\$} & = & \left(I\otimes\$\right)\rho_{\phi},\label{rho-phi-}\end{eqnarray}
 where $\rho_{\phi}=\vert\phi\rangle\langle\phi\vert$. For the channel
$\$$ with a set of Kraus operators $\left\{ K_{i}\right\} $, the
above equation takes the form\begin{eqnarray}
\rho_{\phi,\$} & = & \sum_{i}\left(I\otimes K_{i}\right)\rho_{\phi}\left(I\otimes K_{i}^{\dagger}\right)\label{rho-phi-1}\end{eqnarray}

Recall that,
by transmitting one half of a pure entangled state $\vert\psi\rangle$
through a noisy channel $\Lambda$ results in a mixed state $\rho_{\psi,\Lambda}$
of singlet fraction $F\left(\rho_{\psi,\Lambda}\right)$. Simply maximizing
$F\left(\rho_{\psi,\Lambda}\right)$ over all transmitted pure states
$\vert\psi\rangle$ may not yield the optimal value we are looking
for because it is known \cite{Badziag,Bandyopadhyay,VV-I} that TP-LOCC
can enhance singlet fraction of two qubit states. Thus for a given
$\rho_{\psi,\Lambda}$, the maximum achievable singlet fraction is
defined as \cite{VV-I}\begin{eqnarray}
F^{*}\left(\rho_{\psi,\Lambda}\right) & = & \max_{L}F\left(L\left(\rho_{\psi,\Lambda}\right)\right),\label{max-achievable-fidelity}\end{eqnarray}
 where the maximization is over all TP-LOCC $L$ carried out by Alice
and Bob on their respective qubits. Note that, unlike $F$, which
can increase under TP-LOCC, $F^{*}$ is an entanglement monotone \cite{VV-I}
and can be exactly computed \cite{VV-I} by solving a convex semi-definite
program for any given two-qubit density matrix. Maximizing $F^{*}$
over all transmitted pure states $\vert\psi\rangle$ yields the \emph{optimal
singlet fraction} defined earlier in Eq.$\,$(\ref{optimal-fidelity-0}):\begin{eqnarray}
F\left(\Lambda\right) & = & \max_{\vert\psi\rangle}F^{*}\left(\rho_{\psi,\Lambda}\right).\label{optimal-fidelity}\end{eqnarray}
 It is clear from the above definitions that for any shared pure state
$\vert\psi\rangle$, the following inequalities hold: \begin{equation}
F\left(\Lambda\right)\geq F^{*}\left(\rho_{\psi,\Lambda}\right)\geq F\left(\rho_{\psi,\Lambda}\right).\label{fidelity-inequalities}\end{equation}

Our first result gives an exact formula for the optimal singlet fraction
defined in Eq.$\,$(\ref{optimal-fidelity}) and an explicit protocol by which 
the optimal value can be achieved. We show that for any qubit channel $\Lambda$ there exists an
{}``optimal'' two-qubit pure state $|\psi_{0}\rangle$, not necessarily
maximally entangled, such that all the inequalities in (\ref{fidelity-inequalities})
become equalities.

\begin{thm} The optimal singlet fraction of a qubit channel $\Lambda$
is given by \begin{eqnarray}
F\left(\Lambda\right) & = & \lambda_{\max}\left(\rho_{\Phi^{+},\Lambda}\right),\label{thm-1}\end{eqnarray}
 where $\vert\Phi^{+}\rangle=\frac{1}{\sqrt{2}}\left(\vert00\rangle+\vert11\rangle\right)$,
and $\lambda_{\max}\left(\rho_{\Phi^{+},\Lambda}\right)$ is the maximum
eigenvalue of the density matrix $\rho_{\Phi^{+},\Lambda}$. Moreover,
the following equalities hold: \begin{equation}
F\left(\Lambda\right)=F^{*}\left(\rho_{\psi_{0},\Lambda}\right)=F\left(\rho_{\psi_{0},\Lambda}\right),\label{thm-1-eq-2}\end{equation}
 where $\vert\psi_{0}\rangle$ is the eigenvector corresponding to
the maximum eigenvalue of the density matrix $\rho_{\Phi^{+},\hat{\Lambda}}$\end{thm}

\begin{proof}

We begin by obtaining an exact expression of the maximum pre-processed
singlet fraction. It is defined as\begin{eqnarray}
F_{1}\left(\Lambda\right) & = & \max_{|\psi\rangle}F\left(\rho_{\psi,\Lambda}\right),\label{fidelity-pre-processed}\\
 & = & \max_{|\psi\rangle}\max_{|\Phi\rangle}\langle\Phi|\rho_{\psi,\Lambda}|\Phi\rangle,\label{proof-theorem-1-eq-1}\end{eqnarray}
 where $|\Phi\rangle$ is maximally entangled. Noting that every maximally
entangled state $|\Phi\rangle$ can be written as $U\otimes V|\Phi^{+}\rangle$,
for some $U,V\in SU\left(2\right)$, we can rewrite Eq.$\,$(\ref{proof-theorem-1-eq-1})
as \begin{equation}
F_{1}\left(\Lambda\right)=\max_{|\psi\rangle,U,V}\langle\Phi^{+}|\left(U^{\dagger}\otimes V^{\dagger}\right)\rho_{\psi,\Lambda}\left(U\otimes V\right)|\Phi^{+}\rangle.\label{proof-theorem-1-eq-2}\end{equation}
 
Let, $\rho_{\psi}=|\psi\rangle\langle\psi|$ and $\rho_{\Phi^{+}}=|\Phi^{+}\rangle\langle\Phi^{+}|$.
Using the fact that $(I\otimes V)|\Phi^{+}\rangle=(V^{T}\otimes I)|\Phi^{+}\rangle$,
we now simplify the above equation:  \begin{eqnarray}
F_{1}\left(\Lambda\right) & = & \max_{|\psi\rangle,U,V}\langle\Phi^{+}|\left(U^{\dagger}\otimes V^{\dagger}\right)\rho_{\psi,\Lambda}\left(U\otimes V\right)|\Phi^{+}\rangle\nonumber \\
 & = & \max_{|\psi\rangle,U,V}\langle\Phi^{+}|\left(U^{\dagger}\otimes V^{\dagger}\right)\sum_{i}\left(I\otimes A_{i}\right)\rho_{\psi}\left(I\otimes A_{i}^{\dagger}\right)\left(U\otimes V\right)|\Phi^{+}\rangle\nonumber \\
 & = & \max_{|\psi\rangle,U,V}\langle\psi|\sum_{i}(I\otimes A_{i}^{\dagger})(U\otimes V)\rho_{\Phi^{+}}(U^{\dagger}\otimes V^{\dagger})(I\otimes A_{i})|\psi\rangle\nonumber \\
 & = & \max_{|\psi\rangle,U,V}\langle\psi|\sum_{i}(I\otimes A_{i}^{\dagger})(UV^{T}\otimes I)\rho_{\Phi^{+}}(V^{*}U^{\dagger}\otimes I)(I\otimes A_{i})|\psi\rangle\nonumber \\
 & = & \max_{|\psi\rangle,U,V}\langle\psi|\left(UV^{T}\otimes I\right)\rho_{\Phi^{+},\hat{\Lambda}}\left(V^{*}U^{\dagger}\otimes I\right)|\psi\rangle\nonumber \\
 & = & \max_{|\psi\rangle}\langle\psi|\rho_{\Phi^{+},\hat{\Lambda}}|\psi\rangle,\label{pre-processed-1}\end{eqnarray}

 From the above equation it immediately follows that ,

\begin{equation}
F_{1}\left(\Lambda\right)=F\left(\rho_{\psi_{0},\Lambda}\right)=\lambda_{\max}\left(\rho_{\Phi^{+},\hat{\Lambda}}\right)\label{proof-theorem-1-eq-3}\end{equation}
 where $\lambda_{\max}$ denotes the maximum eigenvalue of $\rho_{\Phi^{+},\hat{\Lambda}}$
and $|\psi_{0}\rangle$ the corresponding eigenvector. Using the result,

\begin{eqnarray}
\lambda_{\max}\left(\rho_{\Phi^{+},\hat{\Lambda}}\right) & = & \lambda_{\max}\left(\rho_{\Phi^{+},\Lambda}\right)\label{proof-theorem-1-eq-4}\end{eqnarray}
 proved in lemma 5(section A of Appendix) , we have therefore proven that \begin{eqnarray}
F\left(\Lambda\right) & \geq & F_{1}\left(\Lambda\right)=\lambda_{\max}\left(\rho_{\Phi^{+},\Lambda}\right)\label{proof-theorem-1-eq-5}\end{eqnarray}
 The following lemma now gives an upper bound on the optimal singlet
fraction $F(\Lambda)$.

\begin{lem} For a qubit channel $\Lambda$ \begin{eqnarray}
F\left(\Lambda\right) & \leq & \lambda_{\max}\left(\rho_{\Phi^{+},\Lambda}\right),\label{lem-1}\end{eqnarray}
 where $\lambda_{\max}\left(\rho_{\Phi^{+},\Lambda}\right)$ denotes the maximum eigenvalue of the density
matrix $\rho_{\Phi^{+},\Lambda}$. \end{lem}

\begin{proof} Recall that by definition, $F\left(\Lambda\right)=\max_{\psi}F^{*}\left(\rho_{\psi,\Lambda}\right)$;
in particular, \begin{equation}
F^{*}\left(\rho_{\psi,\Lambda}\right)=\max_{L}F\left(L\left(\rho_{\psi,\Lambda}\right)\right)=F\left(\rho_{\psi,\Lambda}^{*}\right),\label{max-achievable-fidelity}\end{equation}
 where $\rho_{\psi,\Lambda}^{*}$ is the state obtained from $\rho_{\psi,\Lambda}$
by \emph{optimal} TP-LOCC for a given $\rho_{\psi,\Lambda}$. It was
shown in ref.\cite{VV-I} that the optimal TP-LOCC is an $1$-way LOCC protocol,
where any of the parties apply a state dependent filter. In case of
success the other party does nothing, and in case of failure, Alice
and Bob simply prepare a separable state. We have, therefore, \begin{eqnarray}
\rho_{\psi,\Lambda}^{*} & = & p\rho_{1}+\left(1-p\right)\rho_{s},\label{proof-lemma-1-eq-2}\end{eqnarray}
 where $\rho_{1}=\frac{1}{p}\left(A\otimes I\right)\rho_{\psi,\Lambda}\left(A^{\dagger}\otimes I\right)$
with $A$ being the optimal filter, is the state arising with probability
$p=\mbox{Tr}\left[\left(A^{\dagger}A\otimes I\right)\rho_{\psi,\Lambda}\right]$
when filtering is successful and $\rho_{s}$ is a separable state
which Alice and Bob prepare when the filtering operation is not successful.
$F^{*}$ is given by ( \cite{VV-I}), \begin{eqnarray}
F^{*}\left(\rho_{\psi,\Lambda}\right) & = & F\left(\rho_{\psi,\Lambda}^{*}\right)\nonumber \\
 & = & pF\left(\rho_{1}\right)+\frac{1-p}{2}\label{proof-lemma1-eq-3}\\
 & = & p\langle\Phi^{+}|\rho_{1}|\Phi^{+}\rangle+\frac{1-p}{2}.\label{proof-lemma1-eq-4}\end{eqnarray}
 Observe that the filter is applied at Alice's end, that is, on the
qubit she holds and not on the qubit that was sent through the channel
to Bob. In eqns. (\ref{proof-lemma1-eq-3}) and (\ref{proof-lemma1-eq-4})
, the separable state $\rho_{s}$ is chosen so that $\langle\Phi^{+}|\rho_{s}|\Phi^{+}\rangle=\frac{1}{2}$
and optimality of the filter $A$ implies that $F(\rho_{1})=\langle\Phi^{+}|\rho_{1}|\Phi^{+}\rangle$(if
the latter is not the case we will get another filter unitarily connected
with $A$ which yields higher singlet fraction). We will now show
that $F\left(\rho_{1}\right)\leq\lambda_{\max}\left(\rho_{\Phi^{+},\Lambda}\right)$
. First we observe that  \begin{eqnarray}
F\left(\rho_{1}\right) & = & \frac{1}{p}\langle\Phi^{+}|(A\otimes I)(I\otimes\Lambda)(|\psi\rangle\langle\psi|)\left(A^{\dagger}\otimes I\right)|\Phi^{+}\rangle\nonumber \\
 & = & \frac{1}{p}\langle\Phi^{+}|(I\otimes\Lambda)(A\otimes I)(|\psi\rangle\langle\psi|)\left(A^{\dagger}\otimes I\right)|\Phi^{+}\rangle. \nonumber \\ & & \label{proof-lemma1-eq-5}\end{eqnarray}

On the other hand, because $\Lambda$ is a trace preserving map, we
also observe that \begin{eqnarray}
p & = & \mbox{Tr}\left[(A^{\dagger}A\otimes I)\rho_{\psi,\Lambda}\right]\nonumber \\
 & = & \mbox{Tr}\left[(I\otimes\Lambda)(A^{\dagger}A\otimes I)|\psi\rangle\langle\psi|\right]\nonumber \\
 & = & \mbox{Tr}\left[(A^{\dagger}A\otimes I)\vert\psi\rangle\langle\psi\vert\right]\label{proof-lemma-eq-6}\end{eqnarray}
 .

We thus have $\rho_1= (I\otimes\Lambda)(|\psi^{\prime}\rangle\langle\psi^{\prime}|)$ and from  Eqns.$\,$(\ref{proof-lemma1-eq-5}) and (\ref{proof-lemma-eq-6}) 
we  get \begin{eqnarray}
F\left(\rho_{1}\right) & = & \langle\Phi^{+}|(I\otimes\Lambda)(|\psi^{\prime}\rangle\langle\psi^{\prime}|)|\Phi^{+}\rangle\nonumber \\
 & = & F\left(\rho_{\psi^{\prime},\Lambda}\right),\label{proof-lemma1-eq-7}\end{eqnarray}
 where $\vert\psi^{\prime}\rangle=\frac{1}{\sqrt{q}}\left(A\otimes I\right)\vert\psi\rangle$
is a normalized vector with $q=p=\langle\psi\vert\left(A^{\dagger}A\otimes I\right)\vert\psi\rangle$.
Hence from eqns. (\ref{fidelity-pre-processed}) and (\ref{proof-theorem-1-eq-5})
we have,

\begin{equation}
F\left(\rho_{1}\right)\leq F_{1}\left(\Lambda\right)=\lambda_{\max}\left(\rho_{\Phi^{+},\Lambda}\right).\end{equation}

Thus from Eq.$\,$(\ref{proof-lemma1-eq-4}) we have, \begin{eqnarray}
F^{*}\left(\rho_{\psi,\Lambda}\right) & \leq & p\lambda_{\max}\left(\rho_{\Phi^{+},\Lambda}\right)+\frac{1-p}{2}\nonumber \\
 & \leq & \lambda_{\max}\left(\rho_{\Phi^{+},\Lambda}\right)\label{proof-lemma-1-eq-8b}\end{eqnarray}
 . The last inequality follows from the fact that $\lambda_{\max}\left(\rho_{\Phi^{+},\Lambda}\right)>1/2$
(as the channel is not entanglement breaking, this follows by applying
Lemma 6 (section B of Appendix) on $\rho_{\Phi^{+},\Lambda}$).

Since Inequality (\ref{proof-lemma-1-eq-8b}) holds for any transmitted
pure state $\vert\psi\rangle$, we therefore conclude that \begin{eqnarray}
F\left(\Lambda\right) & \leq & \lambda_{\max}\left(\rho_{\Phi^{+},\Lambda}\right)\label{proof-lemma1-eq-9}\end{eqnarray}
This completes the proof of lemma 1. 
\end{proof}

From Eqs.$\,$(\ref{proof-theorem-1-eq-5}) and (\ref{lem-1}) we
have, $F\left(\Lambda\right)=\lambda_{\max}\left(\rho_{\Phi^{+},\Lambda}\right)$.

Now, as $F\left(\Lambda\right)\geq F^{*}\left(\rho_{\psi_{0},\Lambda}\right)\geq F\left(\rho_{\psi_{0},\Lambda}\right)$
from eqns. (\ref{proof-theorem-1-eq-3}) and (\ref{proof-theorem-1-eq-5})
we have,

\begin{equation}
F\left(\Lambda\right)=F^{*}\left(\rho_{\psi_{0},\Lambda}\right)=F\left(\rho_{\psi_{0},\Lambda}\right)\end{equation}

This completes the proof of theorem 1. 
\end{proof}

What can we say about $|\psi_{0}\rangle$? Evidences so far are mixed:
$\vert\psi_{0}\rangle$ can be either maximally entangled (e.g., for depolarizing
channel \cite{fidelity}) or nonmaximally entangled (e.g., for amplitude
damping channel\cite{BG-2012}), but the answer for a generic qubit channel
is not known. The following result completely characterizes the channels
for which $\vert\psi_{0}\rangle$ is maximally entangled and for which
it is not.

\begin{thm} The state $\vert\psi_{0}\rangle$, as defined in Theorem
1, is maximally entangled if and only if the channel $\Lambda$ is
unital. \end{thm}

\begin{proof} Recall that $|\psi_{0}\rangle$ is
the eigenvector corresponding to the maximum eigenvalue of $\rho_{\Phi^{+},\hat{\Lambda}}$.
Let $|\psi_{0}^{\prime}\rangle$ be the eigenvector corresponding
to the maximum eigenvalue of $\rho_{\Phi^{+}.\Lambda}$. The following
lemma  establishes the correspondence between
the vectors $|\psi_{0}\rangle$ and $|\psi_{0}^{\prime}\rangle$.

\begin{lem} Let $V$ be the swap operator defined by the action $V|\eta\rangle|\chi\rangle=|\chi\rangle|\eta\rangle$.
Then $V|\psi_{0}\rangle^{*}=|\psi_{0}^{\prime}\rangle$. \end{lem}

\begin{proof}

Let us now consider the spectral decomposition of $\rho_{\Phi^{+},\Lambda}$:
Let \begin{eqnarray}
\rho_{\Phi^{+},\Lambda} & = & \sum_{k=0}^{3}p_{k}|\psi_{k}'\rangle\langle\psi_{k}'|,\label{proof-thm-2-eq-1}\end{eqnarray}
 be the spectral decomposition.

 From eqn. (\ref{proof-lem-4})  in the appendix we have, \begin{equation}
\rho_{\Phi^{+},\hat{\Lambda}}=\sum_{k}\lambda_{k}{(V^{\dagger}|\psi_{k}'\rangle\langle\psi_{k}'|V)}^{*}.\label{eq-choi-dual-sd}\end{equation}
 For different values of $k$, ${(V^{\dagger}|\psi_{k}'\rangle)}^{*}$
are orthogonal as $V$ is unitary .

Hence we see that eqn. (\ref{eq-choi-dual-sd}) is in fact a spectral
decomposition of $\rho_{\Phi^{+},\hat{\Lambda}}$ with eigenvectors

\begin{equation}
|\psi_{k}\rangle={(V^{\dagger}|\psi_{k}'\rangle)}^{*}.\label{eq-alpha-alphap}\end{equation}
 The Schmidt coefficients of $|\psi_{k}'\rangle$ are same as
that of $|\psi_{k}\rangle$. The entanglement of $|\psi_{k}'\rangle$
is thus also same as that of $|\psi_{\alpha}\rangle$.

Let $\psi_{0}'$ be the eigenvector corresponding to the maximum eigenvalue
of $\rho_{\Phi^{+},\Lambda}$. We have from eqn. (\ref{eq-alpha-alphap})
, \begin{equation}
|\psi_{0}\rangle={(V^{\dagger}|\psi_{0}'\rangle)}^{*}.\label{eq-psi0-psi0p}\end{equation}
This completes the proof of lemma 2.  
\end{proof}

Therefore, if $|\psi_{0}^{\prime}\rangle$ is maximally entangled,
then so is $|\psi_{0}\rangle$ and vice versa. We will prove the theorem
by showing that $|\psi_{0}^{\prime}\rangle$ is maximally entangled
if and only if $\Lambda$ is unital.

We first show that if $|\psi_{0}^{\prime}\rangle$ is maximally entangled
then $\Lambda$ must be unital. We first note that the Kraus operators
of the channel $\Lambda$ can be obtained from the action of the channel
on the maximally entangled state $|\Phi^{+}\rangle$. 

Now for every $k$, we can write $|\psi_{k}'\rangle$ as\begin{eqnarray}
|\psi_{k}'\rangle & = & \left(I\otimes G_{k}\right)|\Phi^{+}\rangle,\label{proof-thm-2-eq-2}\end{eqnarray}
 where $G_{k}$ is a $2\times2$ complex matrix. It was shown in \cite{fidelity}
that the channel $\Lambda$ can be described in terms of the Kraus
operators $\left\{ \sqrt{p_{k}}G_{k}\right\} $. Noting that (a) $\langle\psi_{i}'|\psi_{j}'\rangle=\delta_{ij}$,
and (b) for any operator $O$, $\langle\Phi^{+}|I\otimes O|\Phi^{+}\rangle=\frac{1}{2}\mbox{Tr}O$,
it follows that the Kraus operators $\left\{ \sqrt{p_{k}}G_{k}\right\} $
are trace orthogonal. That is, \begin{eqnarray}
\mbox{Tr}A_{k}^{\dagger}A_{l} & = & 2\sqrt{p_{k}p_{l}}\delta_{kl},\label{proof-thm-2-eq-3}\end{eqnarray}
 where $A_{k}=\sqrt{p_{k}}G_{k}$. The Kraus operators thus obtained
through the spectral decomposition of $\rho_{\Phi^{+},\Lambda}$ are
trace orthogonal. They also satisfy $\sum A_{k}^{\dagger}A_{k}=I$,
as $\Lambda$ is a TPCP map.

Suppose now the channel $\Lambda$ is non-unital, i.e., $\Lambda\left(I\right)\neq I$.
This implies that \begin{eqnarray}
\sum A_{k}A_{k}^{\dagger} & \neq & I\label{proof-thm-2-eq-4}\end{eqnarray}
 None of our considerations change if we consider a channel $U\circ\Lambda$
with Kraus operators $UA_{k}$ where $U\in SU(2)$ . This is because
the eigenvectors of $\rho_{\Phi^{+},\Lambda}$ and $\rho_{\Phi^{+},U\Lambda}$
are local unitarily connected and eigenvalues are same. Let us now
assume that one of the eigenstates($|\psi_{0}'\rangle$ say) in the
spectral decomposition of ${\rho}_{\Phi^{+},\Lambda}$ in Eq.$\,$(\ref{proof-thm-2-eq-1})
is maximally entangled. This necessarily implies one of the Kraus
operators say, $A_{0}$ is proportional to a unitary. Now because
of the post-processing freedom, without any loss of generality we
can take $A_{0}$ to be $\sqrt{p}I$, with $p\in[0,1]$. Due to trace-orthogonality
{[}Eq.$\,$(\ref{proof-thm-2-eq-3}){]} we will have \begin{equation}
\mbox{Tr}(A_{k})=0,k=1,2,3.\label{traceless-ops}\end{equation}
 We can thus take $A_{k}=\overrightarrow{\alpha_{k}}.\overrightarrow{\sigma}$,
where $\overrightarrow{\alpha_{k}}\in\mathbb{C}^{3}$ and $\overrightarrow{\sigma}=\left\{ \sigma_{x},\sigma_{y},\sigma_{z}\right\} $,
for $k=1,2,3$. On using $\left(\vec{\sigma} \cdot \vec{a}\right) \left(\vec{\sigma} \cdot \vec{b}\right) = (\vec{a} \cdot \vec{b})I + i \vec{\sigma} \cdot (\vec{a} \times \vec{b}) $ the trace preservation condition $\sum A_{k}^{\dagger}A_{k}=I$
now becomes, \begin{equation}
pI+\sum_{k=1}^{3}(\overrightarrow{\alpha_{k}}^{*}\centerdot\overrightarrow{\alpha_{k}})I+i(\overrightarrow{\alpha_{k}}^{*}\times\overrightarrow{\alpha_{k}})\centerdot\overrightarrow{\sigma}=I,\label{trace-preservation}\end{equation}
 from which we obtain,\begin{eqnarray}
p+\sum_{k=1}^{3}(\overrightarrow{\alpha_{k}}^{*}\centerdot\overrightarrow{\alpha_{k}}) & = & 1,\nonumber \\
\sum_{k=1}^{3}\overrightarrow{\alpha_{k}}^{*}\times\overrightarrow{\alpha_{k}} & = & 0.\label{TPCP-2}\end{eqnarray}
 On the other hand the condition for non-unitality {[}Eq.$\,$(\ref{proof-thm-2-eq-4}){]}
of the channel gives us, \begin{equation}
pI+\sum_{k=1}^{3}(\overrightarrow{\alpha_{k}}^{*}\centerdot\overrightarrow{\alpha_{k}})I-i(\overrightarrow{\alpha_{k}}^{*}\times\overrightarrow{\alpha_{k}})\centerdot\overrightarrow{\sigma}\neq I.\label{contradiction}\end{equation}
 which is clearly in contradiction with Eqn.$\,$(\ref{TPCP-2}) .
Thus $\rho_{\Phi^{+},\Lambda}$ cannot have a maximally entangled
eigenvector if $\Lambda$ is non-unital. Hence, $|\psi_{0}'\rangle$
is not maximally entangled. Therefore it follows that if $|\psi_{0}\rangle$
is maximally entangled, then the channel must be unital.

We will now show that if $\Lambda$ is unital then $|\psi_{0}^{\prime}\rangle$
is maximally entangled. In \cite{Ruskai-2002} it was proved that
that for any unital qubit channel $\Lambda$, $\rho_{\Phi^{+},\Lambda}$
is local unitarily connected to the Bell-diagonal state $\sum_{i=0}^{3}p_{i}(I\otimes\sigma_{i})|\Phi^{+}\rangle\langle\Phi^{+}|(I\otimes\sigma_{i})$
with $\sigma_{0}=I$, $1\ge p_{i}\ge0$ and $\sum_{i}p_{i}=1$. It
immediately follows that $|\psi_{0}'\rangle$ is maximally entangled.
This completes the proof of  theorem 2. \end{proof}

\section{ Optimal
singlet fraction and the maximum output negativity} Here we show
that $F\left(\Lambda\right)$ is related to the negativity of the
density matrix $\rho_{\Phi^{+},\Lambda}.$ We first note that an upper
bound on $F^{*}(\rho_{\psi,\Lambda})$ can be given 
in terms of its negativity \cite{F<=1/2} $N\left(\rho_{\psi,\Lambda}\right)$: \begin{equation}
F^{*}\left(\rho_{\psi,\Lambda}\right)\leq\frac{1}{2}\left[1+N\left(\rho_{\psi,\Lambda}\right)\right],\label{eq-negativity}\end{equation}
 where $N\left(\rho_{\psi,\Lambda}\right)=\max\left\{ 0,-2\lambda_{\min}\left(\rho_{\psi,\Lambda}^{\Gamma}\right)\right\} $
and $\rho_{\psi,\Lambda}^{\Gamma}$ is the partially transposed matrix
obtained from $\rho_{\psi,\Lambda}$. Maximizing over all input states
$\vert\psi\rangle$we get, \begin{equation}
F(\Lambda)\leq\frac{1}{2}\left[1+N(\Lambda)\right],\label{negativity-upper-bound}\end{equation}
 where $N\left(\Lambda\right)=\max_{\psi}N\left(\rho_{\psi,\Lambda}\right)$.
An interesting question here is, does the optimal singlet fraction
always reach the above upper bound for all channels $\Lambda?$ In
order to answer this question, we first prove the following: \begin{lem}
For a qubit channel $\Lambda$, the optimal singlet fraction $F\left(\Lambda\right)$
is related to the negativity $N\left(\rho_{\Phi^{+},\Lambda}\right)$
of the state $\rho_{\Phi^{+},\Lambda}$ by the following relation:
\begin{eqnarray}
F\left(\Lambda\right) & = & \frac{1}{2}\left[1+N\left(\rho_{\Phi^{+},\Lambda}\right)\right]\label{optimal-sf-negativity}\end{eqnarray}
 \end{lem}
\begin{proof}
The proof follows by using the formula of negativity, simple
application of Lemma 6 (see section B of Appendix) and Thm 1:

\begin{eqnarray}
\frac{1}{2}\left[1+N\left(\rho_{\Phi^{+},\Lambda}\right)\right] & = & \frac{1}{2}\left[1-2\lambda_{\min}\left(\rho_{\Phi^{+},\Lambda}^{\Gamma}\right)\right]\nonumber \\
 & = & \lambda_{\max}\left(\rho_{\Phi^{+},\Lambda}\right)=F\left(\Lambda\right) \label{eqn-24}\end{eqnarray}
This completes the proof of lemma 3. 
\end{proof}
\vspace{3mm} 
Next we show that  that, $F\left(\Lambda\right)$ does
not reach the upper bound in Eq.$\,$(\ref{negativity-upper-bound})
for all non-unital channels as there are examples for which $N\left(\Lambda\right)>N\left(\rho_{\Phi^{+},\Lambda}\right)$.
Thus, even though the ordering of negativity may change under one-sided
channel action, $I\otimes\Lambda$ the optimal singlet fraction obeys
the bound in Eq.$\,$(\ref{eq-negativity}) for maximally entangled
input. For unital channels however, as the next lemma shows, we have $N\left(\Lambda\right)=N\left(\rho_{\Phi^{+},\Lambda}\right)$.

\begin{lem}
For unital qubit channels we have  $N\left(\Lambda\right)=N\left(\rho_{\Phi^{+},\Lambda}\right)$
\end{lem}

\begin{proof}
The most general two qubit pure state in the Schmidt form is given
by, $|\alpha\rangle=\sqrt{\lambda}|e_{1}f_{1}\rangle+\sqrt{1-\lambda}|e_{2}f_{2}\rangle=(U\otimes V)(\sqrt{\lambda}|00\rangle+\sqrt{(1-\lambda)}|11\rangle)$,
with $\lambda\in[0,1]$ and the $2\times2$ unitary matrices $U$
and $V$ being given by: $U|0\rangle=|e_{1}\rangle,V|0\rangle=|f_{1}\rangle,U|1\rangle=|e_{2}\rangle$
and $V|1\rangle=|f_{2}\rangle$.

For $\lambda\in[0,1]$, let \begin{equation}
W_{\lambda}=\sqrt{\lambda}|0\rangle\langle0|+\sqrt{(1-\lambda)}|1\rangle\langle1|.\label{Wlambda-filter}\end{equation}

Now using the fact that $\Lambda$ is a trace-preserving map it is
easy to show that, \begin{eqnarray}
\rho_{\alpha,\Lambda} & = & (I\otimes\Lambda)|\alpha\rangle\langle\alpha|\nonumber \\
 & = & \frac{(A_{1}\otimes I)\rho_{\Phi^{+},\Lambda}(A_{1}^{\dagger}\otimes I)}{Tr((A_{1}^{\dagger}A_{1}\otimes I)\rho_{\Phi^{+},\Lambda})}, \label{bell-local-filtering}\end{eqnarray}
 with the filter $A_{1}=UW_{\lambda}V^{T}$ .

For a unital channel $\Lambda$ , $\rho_{\Phi^{+},\Lambda}$ is Bell-diagonal
(see proof of theorem 2). In ref. \cite{Ver01} it was shown that negativity
of a Bell-diagonal state cannot be increased by local filtering. Hence, from eqn. (\ref{bell-local-filtering})
for a unital qubit channel $\Lambda$ we have \begin{equation}
N(\Lambda)=N(\rho_{\Phi^{+},\Lambda}).\end{equation}
This completes the proof of lemma 4. 
\end{proof}

\subsubsection{Example of  channel for which $N(\Lambda) > N(\rho_{\Phi^+,\Lambda})$}

Let us consider the amplitude damping channel, with Kraus operators
$K_{0}=\begin{pmatrix}1 & 0\\
0 & \sqrt{1-p}\end{pmatrix}$ and $K_{1}=\begin{pmatrix}0 & \sqrt{p}\\
0 & 0\end{pmatrix}$ with $1\leq p\leq0$ . The channel is non-unital.

It was shown in \cite{BG-2012} that the optimal input state for attaining
optimal singlet fraction of the channel is given by, $|\chi\rangle=\frac{1}{\sqrt{(2-p)}}|00\rangle+\sqrt{\frac{1-p}{2-p}}|11\rangle$.

Using theorem 1 for the amplitude damping channel $\Lambda$ we therefore
get , $F\left(\Lambda\right)=\lambda_{\max}\left(\rho_{\Phi^{+},\Lambda}\right)=F^{*}(\rho_{\chi,\Lambda})=F(\rho_{\chi,\Lambda})$. Now
from eqn. (\ref{eq-negativity})  we get $F^{*}\left(\rho_{\chi,\Lambda}\right)\leq\frac{1}{2}\left[1+N\left(\rho_{\chi,\Lambda}\right)\right]$,
while from lemma 3  we get $F\left(\Lambda\right)=\frac{1}{2}\left[1+N\left(\rho_{\Phi^{+},\Lambda}\right)\right]$.
Hence we must have, $N\left(\rho_{\Phi^{+},\Lambda}\right)\leq N\left(\rho_{\chi,\Lambda}\right)$.

For the amplitude damping channels for input states $|\phi(\lambda)\rangle=\sqrt{\lambda}|00\rangle+\sqrt{(1-\lambda)}|11\rangle$($\lambda\in[0,1]$)
we have, \begin{equation}
N\left(\rho_{\phi(\lambda),\Lambda}\right)=\sqrt{p^{2}(1-\lambda)^{2}+4\lambda(1-\lambda)(1-p)}-(1-\lambda)p.\end{equation}

Thus, \[
N\left(\rho_{\Phi^{+},\Lambda}\right)=\sqrt{\left(\frac{p^{2}}{4}+1-p\right)}-\frac{p}{2}\]
 and, \[
N\left(\rho_{\phi(\frac{1}{2-p}),\Lambda}\right)=\frac{1-p}{2-p}(\sqrt{p^{2}+4}-p).\]

It is easy to see that $N\left(\rho_{\Phi^{+},\Lambda}\right)<N\left(\rho_{\phi(\frac{1}{2-p}),\Lambda}\right)$
for all $1>p>0$ and hence $N\left(\rho_{\Phi^{+},\Lambda}\right)<N(\Lambda)$.

\section{Nonunital channels and maximally entangled input} It is
important to recognize that theorems 1 and 2 put together only prescribes
a method to attain the optimal singlet fraction. It does not, however,
rule out the possibility that the optimal singlet fraction for a nonunital
channel may still be attained by sending part of a maximally entangled
state followed by local post-processing. As it turns out this is not
the case.

\begin{thm} For a nonunital qubit channel $\Lambda$, \begin{eqnarray}
F^{*}\left(\rho_{\Phi^{+},\Lambda}\right) & < & F\left(\Lambda\right)\label{Thm-3}\end{eqnarray}
 \end{thm} \begin{proof} Using the bound in Eq.$\,$(\ref{eq-negativity})
for the  density matrix $\rho_{\Phi^{+},\Lambda}$ we have \begin{eqnarray}
F^{*}\left(\rho_{\Phi^{+},\Lambda}\right) & \leq & \frac{1}{2}\left[1+N\left(\rho_{\Phi^{+},\Lambda}\right)\right].\label{proof-thm3-eq-1}\end{eqnarray}
 It follows from lemma $3$ that to prove theorem $3$ it suffices to show
that for a nonunital channel $\Lambda$, \begin{eqnarray}
F^{*}\left(\rho_{\Phi^{+},\Lambda}\right) & < & \frac{1}{2}\left[1+N\left(\rho_{\Phi^{+},\Lambda}\right)\right].\label{proof-thm-3-eq-1.1}\end{eqnarray}
 As shown in \cite{VV-I}, for any two qubit density matrix $\rho$
the optimal fidelity $F^{*}(\rho)$ can be found by solving the following
convex semidefinite program:\begin{equation}
\mbox{maximize }\:\: F^{*}=\frac{1}{2}-\mbox{Tr}(X\rho^{\Gamma}),\label{proof-thm-3-eq-2}\end{equation}
 under the constraints \begin{eqnarray}
0 & \leq X\leq & I_{4},\label{constraints-1}\\
-\frac{I_{4}}{2} & \leq X^{\Gamma}\leq & \frac{I_{4}}{2},\label{constraints-2}\end{eqnarray}
 with $X^{\Gamma}$ being the partial transpose of $X$. In addition,
the optimal $X$ is known to be of rank one.

The proof is now by contradiction. Suppose that $F^{*}\left(\rho_{\Phi^{+},\Lambda}\right)=\frac{1}{2}\left[1+N\left(\rho_{\Phi^{+},\Lambda}\right)\right]$;
thus to achieve this equality we must necessarily have,\begin{eqnarray}
\frac{1}{2}-\mbox{Tr}(X_{\mbox{opt}}\rho_{\Phi^{+},\Lambda}^{\Gamma}) & = & \frac{1}{2}\left[1+N\left(\rho_{\Phi^{+},\Lambda}\right)\right],\label{proof-thm-3-eq-3}\end{eqnarray}
 from which it follows that\begin{eqnarray}
\mbox{Tr}(X_{\mbox{opt}}\rho_{\Phi^{+},\Lambda}^{\Gamma}) & = & -\frac{N\left(\rho_{\Phi^{+},\Lambda}\right)}{2}\nonumber \\
 & = & \lambda_{\min}\left(\rho_{\Phi^{+},\Lambda}^{\Gamma}\right).\label{proof-thm-3-eq-4}\end{eqnarray}
 Using the facts that $X_{\mbox{opt}}$ is a positive rank one operator
(proved in \cite{VV-I}) and there is only one negative eigenvalue
for $\rho_{\Phi^{+},\Lambda}^{\Gamma}$ (which means $\lambda_{\min}$
is negative), we obtain\begin{equation}
X_{\mbox{opt}}=|\alpha\rangle\langle\alpha|,\label{proof-thm-3-eq-5}\end{equation}
 where $\rho^{\Gamma}|\alpha\rangle=\lambda_{\mbox{min}}(\rho^{\Gamma})|\alpha\rangle$.
Clearly $X_{\mbox{opt}}$ in the above eqn. is of rank one and satisfies
$0\leq X\leq I_{4}$. As eigenvalues of $X$ and $X^{\Gamma}$ are
invariant under local unitaries it is sufficient to take , \begin{equation}
X=\mathbf{P}(\sqrt{\lambda}|00\rangle+\sqrt{(1-\lambda)}|11\rangle),\label{X}\end{equation}
 with $\mathbf{P}(|a\rangle)$ denoting projector on $|a\rangle$ and $\lambda \in (0,1)$. 
.

The spectrum of $X^{\Gamma}$ for $X$ in Eq.$\,$(\ref{X}) is given
by , \begin{equation}
\lambda(X^{\Gamma})=\lambda,(1-\lambda),\pm\sqrt{\lambda(1-\lambda)}.\label{proof-thm-3-eq-6}\end{equation}
 Thus the constraint (\ref{constraints-2}) is only satisfied for
$\lambda=\frac{1}{2}$ , i.e, if $|\alpha\rangle$ is maximally entangled.
Therefore, under the assumption $F^{*}\left(\rho_{\Phi^{+},\Lambda}\right)=\frac{1}{2}\left[1+N\left(\rho_{\Phi^{+},\Lambda}\right)\right]$,
the eigenvector $|\alpha\rangle$ corresponding to the negative eigenvalue
$\lambda_{\min}\left(\rho_{\Phi^{+},\Lambda}^{\Gamma}\right)$ is
maximally entangled.

But then this implies that \begin{eqnarray}
F\left(\rho_{\Phi^{+},\Lambda}\right) & = & \frac{1}{2}\left[1+N\left(\rho_{\Phi^{+},\Lambda}\right)\right]=\mbox{\ensuremath{\lambda}}_{\max}\left(\rho_{\Phi^{+},\Lambda}\right)\label{proof-thm-3-eq-7}\end{eqnarray}
 because for any two qubit entangled density matrix $\sigma$, $F\left(\sigma\right)=\frac{1}{2}\left[1+N\left(\sigma\right)\right]$
if and only if the eigenvector corresponding to the negative eigenvalue
of $\sigma^{\Gamma}$ is maximally entangled \cite{F<=1/2}.The last
equality in eqn. (\ref{proof-thm-3-eq-7}) follows from eqn. \ref{eqn-24}.

Now from theorem 1 we have, \begin{equation}
F\left(\Lambda\right)=F\left(\rho_{\psi_{0},\Lambda}\right)=\lambda_{\max}\left(\rho_{\Phi^{+},\Lambda}\right)\label{proof-thm-3-eq-8}\end{equation}
 where $|\psi_{0}\rangle$ is the eigenvector corresponding to the
maximum eigenvalue of $\rho_{\Phi^{+},\hat{\Lambda}}$. Now from Theorem
2 we know that $|\psi_{0}\rangle$ is necessarily non-maximally entangled
when the channel $\Lambda$ is nonunital. Thus for a nonunital channel
$\Lambda$, \begin{eqnarray}
F(\rho_{\Phi^{+},\Lambda}) & < & F\left(\Lambda\right)=\lambda_{\max}\left(\rho_{\Phi^{+},\Lambda}\right)\label{proof-3-eq-9}\end{eqnarray}
 which contradicts Eq.$\,$(\ref{proof-thm-3-eq-7}). \end{proof} This completes the proof of theorem 3. 

\section{Conclusions} Shared entanglement is a critical resource
for quantum information processing tasks such as quantum teleportation.
Typically, quantum entanglement is shared by sending part of a pure
entangled state through a quantum channel which, in practice is noisy.
This results in mixed entangled states, purity of which is characterized
by singlet fraction. Because faithful implementation of quantum information
processing tasks require near-perfect entangled states (states with
very high purity), a basic question is: What is the optimal singlet
fraction attainable for a single use of a quantum channel $\Lambda$
and trace-preserving local operations? 

In this paper, we obtained an exact expression of the optimal singlet
fraction for a qubit channel and prescribed a protocol to attain the
optimal value. The protocol consists of sending part of a pure entangled
state $\vert\psi_{0}\rangle$ through the channel, where $\vert\psi_{0}\rangle$
is given by the eigenvector corresponding to the maximum eigenvalue
of the density matrix $\rho_{\Phi^{+},\hat{\Lambda}}$ ($\hat{\Lambda}$
is the  channel dual to the qubit channel $\Lambda$). We have also shown that this {}``best'' state
$\vert\psi_{0}\rangle$ is maximally entangled for unital channels
but must be nonmaximally entangled if the channel is nonunital. Interestingly,
we find that in the optimal case no local post-processing is required
even though it is known that TP LOCC can increase singlet fraction
of a density matrix. 

We would also like to mention that recent results \cite{Bruss-et-al-2012,Chuan-et-al-2012,Kay-2012}
have shown that generalized quantum correlations play an essential
role in distribution of entanglement via separable states. In this
setting, the carrier, which always remains separable with the rest
of the system, is transmitted through a noiseless quantum channel,
whereas in practice channels are noisy. We therefore expect our results
to be useful in a more general treatment of the aforementioned scheme
of entanglement distribution involving noisy quantum channels.

\begin{acknowledgments}
R. Pal wants to thank Ajit Iqbal Singh for useful discussions. 
\end{acknowledgments}

\section{Appendix}
\subsection*{A. Technical Lemma }

\begin{lem}  $\lambda_{\max}(\rho_{\Phi^{+},\hat{\Lambda}})=\lambda_{\max}(\rho_{\Phi^{+},\Lambda})$
\end{lem}

Proof. We first obtain a relationship between the states $\rho_{\Phi^{+},\Lambda}$
and $\rho_{\Phi^{+},\hat{\Lambda}}$. Recall that these states are
given by \begin{equation}
\rho_{\Phi^{+},\Lambda}=\sum_{i}(I\otimes A_{i})|\Phi^{+}\rangle\langle\Phi^{+}|(I\otimes A_{i}^{\dagger}).\label{Choi}\end{equation}
 \begin{equation}
\rho_{\Phi^{+},\hat{\Lambda}}=\sum_{i}(I\otimes A_{i}^{\dagger})|\Phi^{+}\rangle\langle\Phi^{+}|(I\otimes A_{i}).\label{Choi-dual}\end{equation}
 Eqn. (\ref{Choi-dual}) can be written as,

\begin{eqnarray}
\rho_{\Phi^{+},\hat{\Lambda}} & = & \sum_{i}((A_{i}^{\dagger})^{T}\otimes I)|\Phi^{+}\rangle\langle\Phi^{+}|(A_{i}^{T}\otimes I)\nonumber \\
\implies\rho_{\Phi^{+},\hat{\Lambda}}^{*} & = & \sum_{i}(A_{i}\otimes I)|\Phi^{+}\rangle\langle\Phi^{+}|(A_{i}^{\dagger}\otimes I),\end{eqnarray}
 where the complex conjugation is taken with respect to the computational
basis $\{|00\rangle,|01\rangle,|10\rangle,|11\rangle\}$. Now using
the SWAP operator V defined by the action $V|ij\rangle=|ji\rangle$,
we have \begin{eqnarray}
(A_{i}\otimes I)|\Phi^{+}\rangle & = & \frac{1}{\sqrt{2}}\sum_{k=0}^{1}A_{i}|k\rangle\otimes|k\rangle\:\:\mbox{ and so,}\nonumber \\
V(A_{i}\otimes I)|\Phi^{+}\rangle & = & \frac{1}{\sqrt{2}}\sum_{k=0}^{1}|k\rangle\otimes A_{i}|k\rangle\nonumber \\
 & = & (I\otimes A_{i})|\Phi^{+}\rangle.\end{eqnarray}
 Hence, \begin{eqnarray}
\rho_{\Phi^{+},\hat{\Lambda}}^{*} & = & V^{\dagger}\rho_{\Phi^{+},\Lambda}V,\nonumber \\
\implies\rho_{\Phi^{+},\hat{\Lambda}} & = & {(V^{\dagger}\rho_{\Phi^{+},\Lambda}V)}^{*}.\label{proof-lem-4}\end{eqnarray}
 From the above equation it therefore follows that \begin{equation}
\lambda_{\max}(\rho_{\Phi^{+},\hat{\Lambda}})=\lambda_{\max}(\rho_{\Phi^{+},\Lambda}).\end{equation}

Note that lemma 5 does not assume that $\Lambda$ is a  qubit channel. Also, from eqn. (\ref{proof-lem-4}) it is clear that $\rho_{\Phi^{+},\hat{\Lambda}}$ is a valid state even for a non-unital channel $\Lambda$ for which the dual channel $\hat{\Lambda}$
is not trace preserving. But we will get unnormalized states if the dual channel acts on one side of some non-maximally entangled states.     

\subsection*{B. Technical Lemma}

\begin{lem}  Let $\sigma_{AB}\in\mathbb{C}^{2}\otimes\mathbb{C}^{2}$
be a bipartite density matrix such that $\mbox{Tr}_{B}\left(\sigma_{AB}\right)=\frac{1}{2}I$.
Then, \begin{eqnarray}
\lambda_{\min}\left(\sigma_{AB}^{\Gamma}\right)+\lambda_{\max}\left(\sigma_{AB}\right) & = & \frac{1}{2}\label{lem-2-eq-1}\end{eqnarray}
 where $\lambda_{\min}\left(X\right)$ and $\lambda_{\max}\left(X\right)$
denote the minimum and maximum eigenvalue of $X\in\left\{ \sigma_{AB},\sigma_{AB}^{\Gamma}\right\} $
and $\Gamma$ denotes partial transposition. \end{lem}

Proof. Let $\sigma_{AB}\in\mathbb{C}^{2}\otimes\mathbb{C}^{2}$ be
a bipartite density matrix such that $\mbox{Tr}_{B}\left(\sigma_{AB}\right)=\frac{1}{2}I$.
From the Choi-Jamiolkowski isomorphism (\cite{JAM72}, \cite{CH75})
we have that $\sigma_{AB}$ can be written as , \[
\sigma_{AB}=(I\otimes\Lambda)\left({|\Phi^{+}\rangle_{AB}\langle\Phi^{+}|}\right),\]
 where $\Lambda$ is trace preserving completely positive map(TPCP),
mapping $\mathcal{B}(\mathcal{C}^{2})$ to itself.

In \cite{Ruskai-2002} it was shown that any such map $\Lambda$ can
be written as, \begin{equation}
\Lambda(\rho)=U_{1}\circ\Lambda'\circ U_{2}(\rho)\label{map-composition}\end{equation}
 with $\Lambda'$ being a canonical TPCP map and $U_{1}$ and $U_{2}$
being unitary maps. If $\rho=\frac{1}{2}(I+x\sigma_{1}+y\sigma_{2}+z\sigma_{3})$
and $\rho'=\Lambda'(\rho)=\frac{1}{2}(I+x'\sigma_{1}+y'\sigma_{2}+z'\sigma_{3})$
then in the Bloch sphere representation the map $\Lambda'$ is given
by, \begin{equation}
\left[\begin{array}{c}
1\\
x'\\
y'\\
z'\end{array}\right]=\begin{bmatrix}1 & 0 & 0 & 0\\
t_{1} & \lambda_{1} & 0 & 0\\
t_{2} & 0 & \lambda_{2} & 0\\
t_{3} & 0 & 0 & \lambda_{3}\end{bmatrix}\left[\begin{array}{c}
1\\
x\\
y\\
z\end{array}\right],\label{map-canonical}\end{equation}
 with $t_{i}$ and $\lambda_{i}$ being real for all $i$.

Now as local unitaries do not affect the eigenvalues of $\sigma_{AB}$
or $\sigma_{AB}^{\Gamma}$ , for the rest of the proof we can focus
on $(I\otimes\Lambda')(|\Phi^{+}\rangle\langle\Phi^{+}|)=\rho_{\Phi^+,\Lambda'}$
with the map $\Lambda'$ given by eqn. (\ref{map-canonical}) . We
have,

\begin{equation}
\rho_{\Phi^+,\Lambda'}=\frac{1}{2}\begin{bmatrix}a & b & 0 & d\\
b^* & (1-a) & f & 0\\
0 & f & c & b\\
d & 0 & b^* & (1-c)\end{bmatrix}\end{equation}
 , with $a=\frac{1+t_{3}+\lambda_{3}}{2}$, $b=\frac{t_{1}-it_{2}}{2}$,
$d=\frac{(\lambda_{1}+\lambda_{2})}{2}$, $f=\frac{(\lambda_{1}-\lambda_{2})}{2}$,
$c=\frac{(1+t_{3}-\lambda_{3})}{2}$ . Now complete positivity of
$\Lambda'$ implies positivity of $\rho_{\Phi^+,\Lambda'}$ and hence the
spectrum of $\rho_{\Phi^+,\Lambda'}$ is same as that of $\rho_{\Phi^+,\Lambda'}^{*}$
. Now the eigenvalue equation of $\rho_{\Phi^+,\Lambda'}^{*}$ is \begin{equation}
\begin{vmatrix}(\frac{a}{2}-\lambda) & \frac{b^{*}}{2} & 0 & \frac{d}{2}\\
\frac{b}{2} & (\frac{1-a}{2}-\lambda) & \frac{f}{2} & 0\\
0 & \frac{f}{2} & (\frac{c}{2}-\lambda) & \frac{b^{*}}{2}\\
\frac{d}{2} & 0 & \frac{b}{2} & (\frac{(1-c)}{2}-\lambda)\end{vmatrix}=0.\label{eig-rho*}\end{equation}

Now, the partial transpose w.r.t first party of $\rho_{\Phi^+,\Lambda'}$
is given by, \begin{equation}
\rho_{\Phi^+,\Lambda'}^{\Gamma}=\frac{1}{2}\begin{bmatrix}a & b & 0 & f\\
b^{*} & (1-a) & d & 0\\
0 & d & c & b\\
f & 0 & b^{*} & (1-c)\end{bmatrix}.\end{equation}
 The eigenvalue equation of $\rho_{\Phi^+,\Lambda'}^{\Gamma}$ is given by,
\begin{equation}
\begin{vmatrix}(\frac{a}{2}-\lambda) & \frac{b}{2} & 0 & \frac{f}{2}\\
\frac{b^{*}}{2} & (\frac{(1-a)}{2}-\lambda) & \frac{d}{2} & 0\\
0 & \frac{d}{2} & (\frac{c}{2}-\lambda) & \frac{b}{2}\\
\frac{f}{2} & 0 & \frac{b^{*}}{2} & (\frac{(1-c)}{2}-\lambda)\end{vmatrix}=0.\label{eig-rhop}\end{equation}

Replacing $\lambda$ by $(\frac{1}{2}-\lambda')$, in eqn. (\ref{eig-rhop})
we have,

\begin{equation}
\begin{vmatrix}-(\frac{(1-a)}{2}-\lambda') & \frac{b}{2} & 0 & \frac{f}{2}\\
\frac{b^{*}}{2} & -(\frac{a}{2}-\lambda') & \frac{d}{2} & 0\\
0 & \frac{d}{2} & -(\frac{(1-c)}{2}-\lambda') & \frac{b}{2}\\
\frac{f}{2} & 0 & \frac{b^{*}}{2} & -(\frac{c}{2}-\lambda')\end{vmatrix}=0.\label{eig-rhorepl}\end{equation}

In eqn. (\ref{eig-rhorepl}) performing the interchanges, column 1 $\Leftrightarrow$
column 2 and column 3 $\Leftrightarrow$ column 4 we have, \begin{equation}
\begin{vmatrix}\frac{b}{2} & -(\frac{(1-a)}{2}-\lambda') & \frac{f}{2} & 0\\
-(\frac{a}{2}-\lambda') & \frac{b^{*}}{2} & 0 & \frac{d}{2}\\
\frac{d}{2} & 0 & \frac{b}{2} & -(\frac{(1-c)}{2}-\lambda')\\
0 & \frac{f}{2} & -(\frac{c}{2}-\lambda') & \frac{b^{*}}{2}\end{vmatrix}=0.\label{eig-columnint}\end{equation}

In eqn. (\ref{eig-columnint}) performing the interchanges,  row 1 $\Leftrightarrow$
row 2 and row 3 $\Leftrightarrow$ row 4 we have, \begin{equation}
\begin{vmatrix}-(\frac{a}{2}-\lambda') & \frac{b^{*}}{2} & 0 & \frac{d}{2}\\
\frac{b}{2} & -(\frac{(1-a)}{2}-\lambda') & \frac{f}{2} & 0\\
0 & \frac{f}{2} & -(\frac{c}{2}-\lambda') & \frac{b^{*}}{2}\\
\frac{d}{2} & 0 & \frac{b}{2} & -(\frac{(1-c)}{2}-\lambda')\end{vmatrix}=0.\label{eig-rowint}\end{equation}
 Now multiplying the 1st row by -1, 2nd column by -1, 3rd row by -1
and 4th column by -1 successively in eqn. (\ref{eig-rowint}) we get
back eqn. (\ref{eig-rho*}) . Thus if eigenvalues of $\rho_{\Phi^+,\Lambda'}$
are $\lambda_{i}$ with $i=1,2,3,4$, that of $\rho_{\Phi^+,\Lambda'}^{\Gamma}$
are $(\frac{1}{2}-\lambda_{i})$. Thus we have, \begin{eqnarray}
\lambda_{\min}(\rho_{\Phi^+,\Lambda'}^{\Gamma}) & = & \frac{1}{2}-\lambda_{\max}(\rho_{\Phi^+,\Lambda'})\nonumber \\
\Rightarrow\lambda_{\min}(\rho_{\Phi^+,\Lambda'}^{\Gamma})+\lambda_{\max}(\rho_{\Phi^+,\Lambda'}) & = & \frac{1}{2}\nonumber \\
\Rightarrow\lambda_{\min}(\sigma_{AB}^{\Gamma})+\lambda_{\max}(\sigma_{AB}) & = & \frac{1}{2}.\end{eqnarray}

\end{document}